\documentclass[sigconf]{acmart}

\usepackage{booktabs} 

\usepackage{amsmath,amssymb}
\usepackage{cases}
\usepackage{epsfig}
\usepackage{graphics,graphicx}
\usepackage{float}
\usepackage{subfigure}
\usepackage[vlined,boxed,linesnumbered]{algorithm2e}

\usepackage{tikz}
\usetikzlibrary{shapes,arrows,automata,positioning}

\makeatletter
\newcommand{\removelatexerror}{\let\@latex@error\@gobble}
\makeatother

\begin{document}

\title{On The Equivalence of Tries and Dendrograms - Efficient Hierarchical Clustering of Traffic Data}

\author{Chia-Tung Kuo}
\affiliation{
  \institution{University of California, Davis}
}
\email{tomkuo@ucdavis.edu}

\author{Ian Davidson}
\affiliation{
  \institution{University of California, Davis}
}
\email{davidson@cs.ucdavis.edu}

\begin{abstract}
The widespread use of GPS-enabled devices  generates voluminous and continuous amounts of traffic data but analyzing such data for interpretable and actionable insights poses challenges. A hierarchical clustering of the trips has many uses such as discovering shortest paths, common routes and often traversed areas.  
However, hierarchical clustering typically has time complexity of $O(n^2 \log n)$ where $n$ is the number of instances, and is difficult to scale to large data sets associated with GPS data. Furthermore, incremental hierarchical clustering is still a developing area.
Prefix trees (also called tries) can be efficiently constructed and updated in linear time (in $n$).
We show how a specially constructed trie can compactly store the trips and further show this trie is equivalent to a dendrogram that would have been built by classic agglomerative hierarchical algorithms using a specific distance metric. This allows creating hierarchical clusterings of  GPS trip data and updating this hierarchy in linear time.
We demonstrate the usefulness of our proposed approach on a real world data set of half a million taxis' GPS traces, well beyond the capabilities of agglomerative clustering methods. Our work is not limited to trip data and can be  used with other data with a string representation.
\end{abstract}

%
%

\maketitle

\section{Introduction}
\label{sec:intro}

Location tracking devices have become widely popular over the last decade and this has enabled the collection of large amounts of spatial temporal trip data.
Given a collection of such trip data, clustering is often a natural start to explore the general properties of the data and among clustering methods, hierarchical clustering is well suited as it provides a set of groupings at different levels where each grouping at one level is a refinement of the groupings at the previous levels. 
This dendrogram structure can also naturally represent the evolutionary/temporal nature of trip data where each level in the dendrogram corresponds to a particular time.
However, one significant drawback of the standard agglomerative hierarchical clustering is its $O(n^2 \log n)$ run time, which is prohibitive when the number of instances, $n$ is large. For example, in our experiments our data set has nearly half a million trips. Classic agglomerative algorithms would take days or even weeks to build a dendrogram on such large data set. 

In this paper we consider an alternative way to efficiently construct a dendrogram of the trip data without the long computation time.
We achieve this by first converting each trip to a trajectory string and then building a prefix tree from the trip-strings. We then formally show the equivalence between a prefix tree and a dendrogram by showing that the prefix tree we create would have been built by a classic agglomerative method using a specific distance metric on the strings. This result is not trivial as the prefix tree is created top-down and the dendrogram bottom-up.

\textbf{Creating Trajectory Strings from GPS Data.}
We discretize both the spatial and temporal dimensions with respective pre-defined resolutions, effectively converting each spatial location to a unique symbol. For example, in our experiments we disceretize the San Francisco Bay area into a 100 $\times$ 100 grid so our alphabet contains 10,000 symbols. We can then naturally represent a trip as a sequence (i.e. string) of the discretized regions (symbols). The symbol at position $i$ in the string represents the location of the trip at time step $i$ as shown in Figure \ref{fig:build_string}.


\textbf{Creating Trip Tries.} A prefix tree is a tree structure built from strings where each path from the root to any node corresponds to a unique string prefix and is commonly used for indexing and retrievals of text and symbolic data. 
A prefix tree (trie) can be constructed in linear time to both number of trips, $n$ and maximum number of discretized time steps, $l$. 
An example of such a tree is shown in Figure \ref{fig:tree_example}. 

\textbf{Uses of Trip Tries.} A \emph{trip trie} is not only a hierarchical clustering (as we shall see) but has other uses.
For example, easy to understand visualizations of a collection of trips such as heat maps (see Figures \ref{fig:heatmap} and \ref{fig:loc_occur_order}) can be created from a trip trie; and trip tries constructed from different collections of trips can be compared (i.e. Table \ref{tab:tree_stats}).
Tries have many useful properties such as the ability to efficiently compute Levenshtein distances and we describe uses such as creating more robust clusters using these properties.
\emph{Though tries are commonly used in the database literature for tasks such as retrieval and indexing, to our knowledge they have not been used for the purposes we outlined in this paper.}

\textbf{Uses Beyond GPS Trip Data.} In this paper we have focused on GPS trip data as the application domain is important and has readily available public data. However, our work is applicable in other domains where the data represents behavior over time such as settings where some categorical event (a symbol) occurs over time (the position of the symbol). In our earlier work \cite{davidson2012} we modeled behavioral data as these event strings but other applications exist in areas such as computer network traffic where each location is an IP address.

Our contributions can be summarized as follows.
\begin{itemize}
\item We provide a novel way to organize trip data into symbolic data and then a prefix tree/trie (see section \ref{sec:approach}). Tries can be built and updated in linear time to the number of instances and alphabet size.
\item We derive the equivalence between a prefix tree and a dendrogram and verify it empirically (see Theorem \ref{thm:metric} and derivation in section \ref{sec:property}).
\item We discuss extensions of our work including uses  beyond hierarchical clustering such as outlier detection (see section \ref{sec:other}).
\item We demonstrate the usefulness of our dendrogram in illustrating interesting insights of trip data on a real data set of GPS traces of taxis (see section \ref{sec:experiment}).
\end{itemize}

Our paper is organized as follows. We describe the steps to create string representations and a trip trie in section \ref{sec:approach}, which is then followed by a proof of its equivalence to standard agglomerative hierarchical clustering in section \ref{sec:property}. We further discuss other different ways this trip trie can be used in section \ref{sec:other}. In section \ref{sec:experiment} we evaluate our approach on a real world dataset of taxis' GPS traces and demonstrate the usefulness of our approach in obtaining insights on the traffic dynamics. We discuss related work in section \ref{sec:related} and conclude our paper in section \ref{sec:conclusion}.

\begin{figure}[h]
\centering
\includegraphics[width=0.35\textwidth]{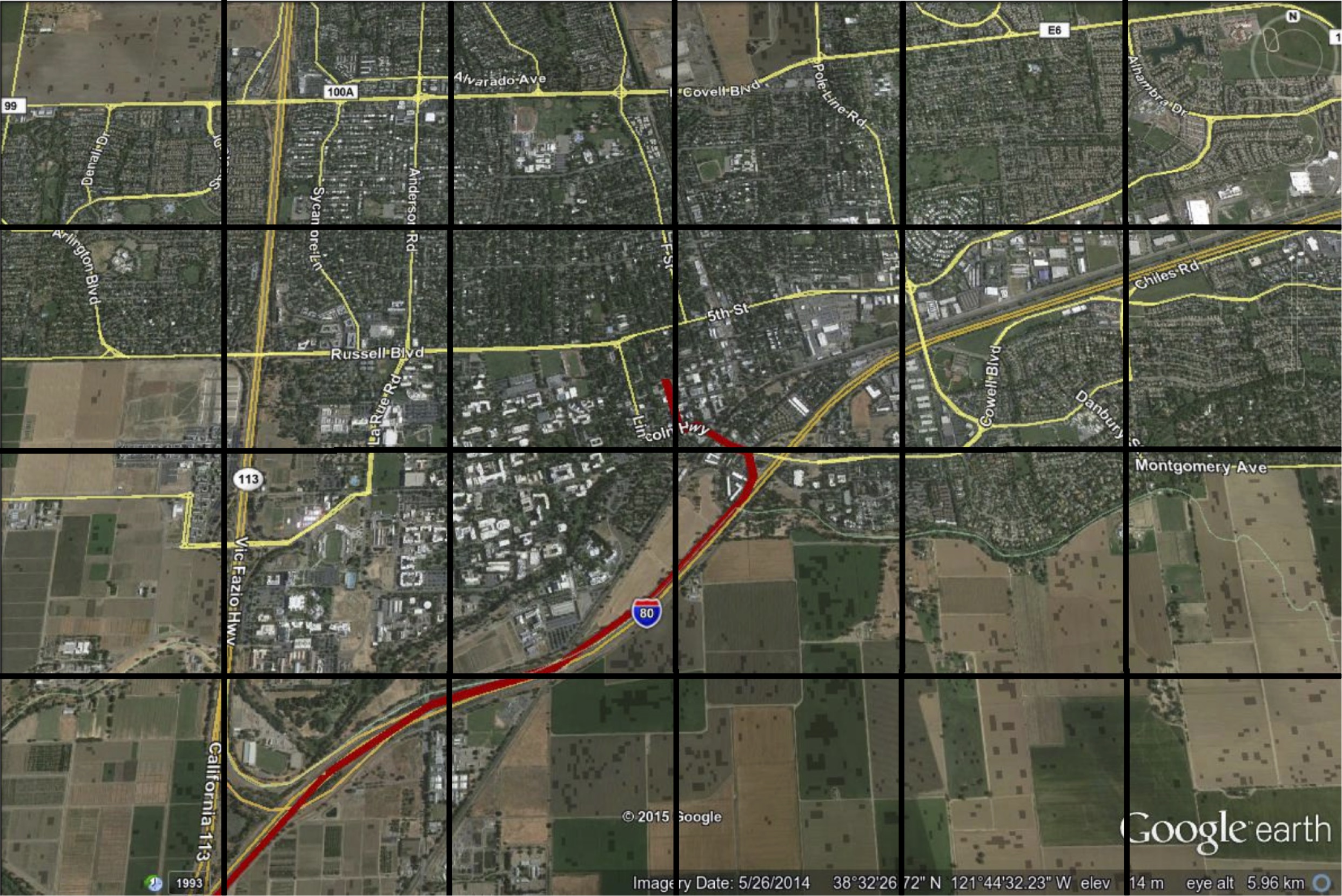}
\caption{
An example that shows the construction of a trajectory string from the given spatial grids (in black) and temporal resolution. 
Assume symbols are assigned to the grids such that the grids in the top row are $z_1, z_2, \dots, z_6$ and the grids in the second row are $z_7, z_8, \dots, z_{12}$, and so on where the last symbol is $z_{24}$ the bottom right grid.
If the starting location is $z_9$, the trip (marked in red) can be represented as $z_{9}z_{10}z_{16}z_{15}z_{21}z_{20}$.}
\label{fig:build_string}
\end{figure}

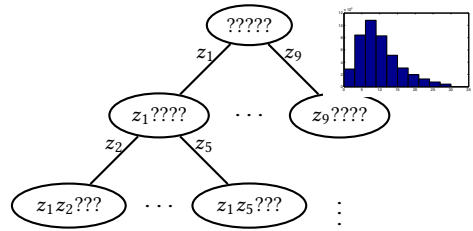
\begin{figure}[h]
\centering
\begin{small}
\begin{tikzpicture}[%
       >=stealth,
       blue/.style = {draw, circle, minimum size = 4.5mm, color = blue},
       green/.style = {draw, circle, minimum size = 4.5mm, color = green},
	   inode/.style = {draw, ellipse, color = black, minimum size = 4.5mm},
	   dots/.style = {color = black, minimum size = 4.5mm},
	   anno/.style = {color = black, minimum size = 4.5mm},
       node distance=12mm,
       thick,
       on grid,
       auto
     ]
     
     \node[inode] (0) {$?????$};
     \node[dots] (1) [below of = 0] {$\ldots$};
     \node[inode] (2) [left of = 1] {$z_1????$};
     \node[inode] (3) [right of = 1] {$z_{9}????$}; 
	 \node[anno] (8) [above right of = 3, node distance = 12mm] {\includegraphics[width=0.12\textwidth]{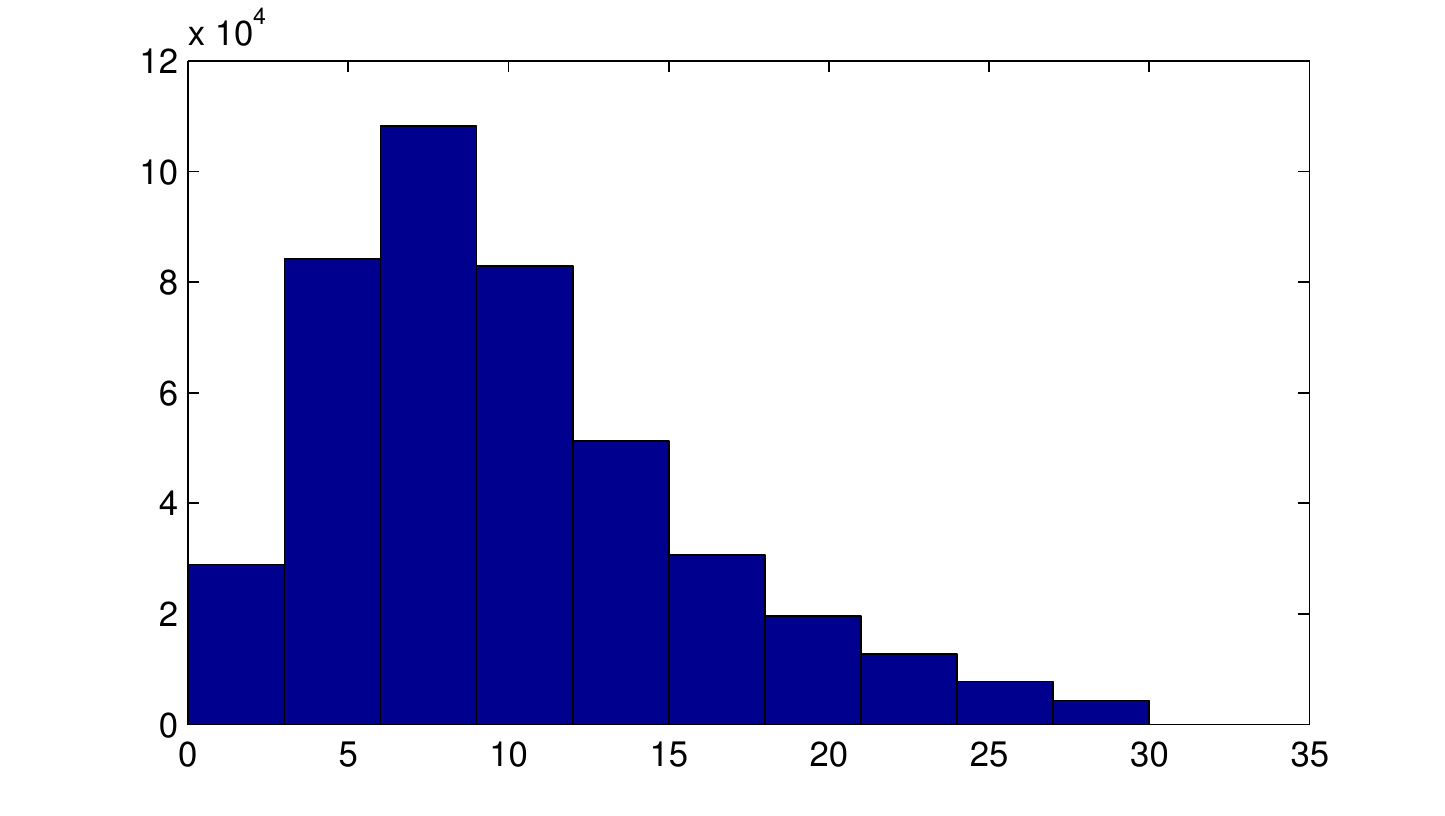}};     
	 \node[dots] (4) [below of = 2] {$\ldots$};
	 
	 \node[inode] (5) [below of = 1] {$z_1z_5???$};
	 \node[dots] (6) [below of = 3] {$\vdots$};     
     
     \node[inode] (7) [left of = 4] {$z_1 z_2???$};

     \path[] (0) edge[above] node {$z_1$} (2);
	 \path[] (0) edge[above] node {$z_9$} (3);  
	 \path[] (2) edge[above] node {$z_2$} (7);    
	 \path[] (2) edge[above] node {$z_5$} (5);
\end{tikzpicture}
\end{small}
\caption{
An example of a trip trie. Each node corresponds to a prefix which matches a subset of the trips with the symbol, ``?" meaning unknown as yet. The root of the tree contains all unknown and denotes the beginning of a trip, the next level time step 1 and so on. The user can store additional relevant data about the trips at each node, such as the trip durations, as shown in the histogram.}
\label{fig:tree_example}
\end{figure}

\section{Creating Trip Trie from GPS Data}
\label{sec:approach}
Here we describe the setting upon which our algorithm is built and detail the steps of constructing the trie structure. We assume our data is composed of trips where each trip consists of 
a sequence of GPS-located spatial temporal points $(x, y, t)$ where $x, y$ are the longitude and latitude, respectively, and $t$ is the time stamp when this sample is recorded. Note here we assume the time stamps are \emph{synchronized/identical} for each trip. In reality GPS tracks at irregular time intervals and later in experiments we will use interpolation/extrapolation.
Our approach consists of three major steps, shown in corresponding order in Figure \ref{subfig:map2symbol}, \ref{subfig:build-string} and \ref{subfig:trie}.

\begin{enumerate}
\item \textbf{Discretization of the geographic space}: This step breaks the modeled space into a finite set of distinct non-overlapping regions $Z = \{z_1, z_2, \dots, z_s\}$ which we will use as symbols in an alphabet (see Figure \ref{fig:build_string}). 
This preprocessing is carried out before the major algorithm is applied, just like most work in trajectory mining \cite{Giannotti2007} \cite{chattopadhyay2013joint}.
In our experiment we use equal-sized rectangular grids over the modeled space; this allows straightforward mapping between actual spatial coordinates to the symbols; however our method can be used with any shaped regions.
It is worth noting that the actual number of regions with activity is typically much smaller than the possible number of grid cells due to physical presence of roads. This is a desirable property allowing large geographic areas to be efficiently represented. For example, in our experiments, though we discretize the San Francisco Bay area into 10,000 cells, less than 20\% of them see any activity.

\item \textbf{Build trajectory strings}: In this step we build a trajectory string for each trip using symbols in $Z$. Note that we regard the beginning of each trip as time 0 and each position of the string indicates the location of the object at a particular time. For example, if we decide the temporal resolution is $t_r$, then a trajectory string of $z_4 z_5 z_8$ records the information that this trip starts (i.e. at its time 0) at region $z_4$, goes through region $z_5$ at time $t_r$ and ends at region $z_8$ at time $2 t_r$.  The temporal resolution is typically given by the devices' capture rate but as mentioned before we use a single resolution and construct strings accordingly through interpolation/extrapolation.
This encoding would produce a forest of tries as the initial starting locations may differ. However we can instead create a symbol for an artifical common starting point that occurs before the trips start. This allows a single trie to be created as there is a common root. 
%
Figure \ref{fig:build_string} gives an illustrative example of the result of such construction.

\item \textbf{Construct trie}: 
Once each trip is converted into a string as above (i.e. sequence of regions), we can construct a trie of these strings. It helps to note that position $i$ in a string indicates location at time step $i$ and hence the nodes at level $i$ (from top down) in the trie represent trips whose first $i$ positions are the same.
Figure \ref{fig:tree_example} shows an example of the resulting trie and what the nodes represent at each level.
In our experiments we implement simpler tries in MATLAB since it provides simple built-in functions and thus an easy way to reproduce the results. 
In more sophisticated implementation the user can store useful relevant information about the trips at each node such as the distribution of the trip durations shown in Figure \ref{fig:tree_example}.
We present our pseudo-code in Figure \ref{fig:overall}. 
\end{enumerate}

%
%

\begin{figure}[!h]
\begin{small}
\removelatexerror
\subfigure[Map longitude/latitude to symbol]{\label{subfig:map2symbol}
	\begin{algorithm}[H]
	\KwIn{Longitude/latitude coordinate $(x, y)$, bottom-left corner and top-right corner $(x_{min}, y_{min})$, $(x_{max}, y_{max})$, and grid size $(n_r, n_c)$}
	\KwOut{Region symbol $z$}
	Grid width $w \leftarrow (x_{max}-x_{min})/n_c$ \;
	Grid height $h \leftarrow (y_{max}-y_{min})/n_r$ \;	
	Longitude index $i_x \leftarrow \lceil (x-x_{min})/w \rceil$ \; 
	Latitude index $i_y \leftarrow \lceil (y-y_{min})/h \rceil$ \;
	Ordinal region index $i \leftarrow (i_y - 1) n_c + i_x$ \;
	$z \leftarrow $ symbol $z_i$ \; 	
	\vspace{0.05in}
	\end{algorithm}
}
\subfigure[Build trajectory string for a trip]{\label{subfig:build-string}
	\begin{algorithm}[H]
	\KwIn{A trip longitude/latitude sequence $(x_1, y_1), \dots, (x_k, y_k)$}
	\KwOut{String representation of the trip $s$}
	\vspace{0.05in}
	Initialize empty string buffer $s$\;
	\For {$i \leftarrow 1$ to $k$}{
		$z \leftarrow $ map $(x_i, y_i)$ to symbol using Algorithm \ref{alg:map}\;
		Append $z$ to $s$\;
	}
	\end{algorithm}
}
\subfigure[Construct trip trie]{\label{subfig:trie}
	\begin{algorithm}[H]
	\KwIn{Collection of strings $S = \{s_1, s_2, \dots, s_n\}$}
	\KwOut{Trie $T$}
	\vspace{0.05in}
	Initialize $T \leftarrow$ empty array of structs consisting of two fields: \textit{prefix} and 	\textit{index}\; 
	\For {$k \leftarrow 1$ to $\texttt{length\_of\_longest\_string}$}{
		$T(k).prefix \leftarrow $ unique strings in $S_k$\;
		$T(k).index \leftarrow $ mapping from index in $S_k$ to $T(k).prefix$\;
	}
	\end{algorithm}
}
\end{small}
\caption{
Pseudo-code for the overall construction of the trie.
In Algorithm \ref{subfig:map2symbol}, we convert the 2D grids to symbols as exemplified in Figure \ref{fig:build_string}.
In Algorithm \ref{subfig:trie}, the computations within the for loop (i.e. lines 3 and 4) can be carried out using the built-in primitive \texttt{unique}() in MATLAB and $S_k$ is defined to be the same collection of strings $S$ but with each string truncated at position $k$. }
\label{fig:overall}
\end{figure}

\section{On the Equivalence Between Trip Tries and Hierarchical Clustering}
\label{sec:property}
Here we make the claim that our top-down trip trie can be efficiently constructed in linear time and that it is identical to the result of standard bottom-up agglomerative hierarchical clustering with the following metric on pairs of strings.
\begin{equation}
d(s, s') = \sum_{i=1}^{l} 2^{(l-i)} \mathbb{I}[s^{(i)} \neq s'^{(i)}]
\label{equ:metric}
\end{equation}
where $\mathbb{I}$ is the 0/1 indicator function. This metric can be viewed as a weighted Hamming distance with the weight diminishing exponentially with the symbol position.

We divide this section into several parts. Firstly, we discuss why such a distance metric is useful, its interpretation and implication. We then prove that our method of constructing trie will generate this clustering and finally we present a brief complexity analysis.

\subsection{Interpretation of Dendrogram and Uses of Clusters}
Our string representation of trips effectively takes spatially and temporally irregular data and converts them into a string. The strings have a natural interpretation: the symbol at position $i$ is where the trip was at time $i$. Our string distance function above then effectively says two trips are more similar if they are \emph{initially} in the same locations.
This string representation also allows our dendrogram's levels to naturally explain the evolution of a trip.
We can then interpret the clusters in the following observations:

\begin{itemize}
\item The first level of the dendrogram contains a clustering of the trips based on their starting locations, the next level a refined clustering based on their starting locations \textbf{and} locations at time step $1$ and so on. 
\item Each path from the root to any node represents a cluster of common trips. 
\end{itemize}

Understanding what these clusters represent is critical to understanding how they can be used. Whereas in our earlier work \cite{kuo2015} a cluster of trips represented trips which started and ended in the same location/time, here we consider the \emph{trajectory} and the \emph{entire duration} taken to complete the trip. Therefore it is possible (and can be desirable in some contexts) that if two types of trips have the same start and end locations but are at different paces or slightly different routes (i.e. because some are completed during rush hour and others at night) they will appear in \textbf{different} clusters. 
Such type of clustering can be used in a variety of settings as follows:

\begin{itemize}
\item
\textbf{Next Movement Prediction.} A current trip can be quickly mapped to a node in the trie and the children of the node determine the likely next locations in the time step. 

\item \textbf{Diversity Route Understanding.} Consider trips between a start and an end locations (symbols). Such trips could appear multiple times in the tree (i.e. different nodes) due to different routes and different travel times. Counting how frequent this combination occurs gives a measure of diversity for the pair of start/end locations.

\item \textbf{Common Ancestors.} Consider two nodes. Their common ancestor represents a bifurication point for these trips. Locations which appear frequently in the tree are therefore naturally hub locations.
\end{itemize}

If these nuances between trips speed and trajectories producing different clusters is undesirable, then we propose a method to post-process the dendrogram (see section \ref{sec:other}). In that section we describe how efficient calculation of the Levenshtein distance between strings using a prefix tree allows us to combine clusters in the dendrogram to alleviate these nuances. 

\subsection{A Proof of Equivalence Between Dendrograms \& Tries}
\label{subsec:dendrogram}
Here we show that our trie can be naturally defined by a (threshold) dendrogram\footnote{We use a \emph{threshold} dendrogram to maintain only the ordering of the merges. A \emph{proximity} dendrogram also records the smallest distance for which a merge occurs (see \cite{Jain1988}).} that maps natural numbers to partitions of our finite set of \textit{distinct}\footnote{The assumption of distinct set makes the following definitions and explanation cleaner. In our experiments we don't need to discard duplicate trips as we record at each node (i.e. prefix) which trips have this prefix.}
 trajectory strings $\Sigma = \{s_1, s_2, \dots, s_n\}$. 
We append null symbols, $\emptyset$, to shorter strings so that all strings will have equal length, $l$. Accordingly a string $s$ can be written in its symbols as $s^{(1)} \dots s^{(l)}$.
Our aim here is to claim that our trip trie is identical to the result of single linkage hierarchical clustering on the trajectory strings with a specific metric in equation \ref{equ:metric}, which would require $O(n^2 \log n)$ computation time if done directly.
We will verify this claim empirically in the experimental section.
We achieve this aim with 3 steps.
\begin{enumerate}
\item Define equivalence relations from the prefixes such that an equivalence class forms a single node in one level of the trie and all equivalence classes form all clusters at one level in the trie.
\item Define our trip trie as a dendrogram; that is the function mapping from natural numbers (i.e. levels) to the set partitions of the strings ($\Sigma$).
\item Define a metric on the strings so that single linkage hierarchical clustering outputs our exact dendrogram above.
\end{enumerate}
We will follow the notations and definitions used in \cite{Carlsson2010} throughout the discussion.

\textbf{Equivalence Classes} We define equivalence relations $r_k$ on $\Sigma$ for each integer $k$ with $0 \leq k \leq l$: $s \sim_{r_k} s'$ 
if and only if $s^{(i)} = s'^{(i)} ~~ \forall i \leq l-k$; in other words, strings $s$ and $s'$ are equivalent if they share \textit{common prefix} of length $l-k$; it should be clear from the construction of a trie that all trips going through the same node at level $l-i$ are equivalent under relation $r_i$.
It is straightforward to verify that these are indeed equivalence relations.

\textbf{Dendrogram.} 
We formally define our dendrogram $\theta: \mathbb{N} \rightarrow P(\Sigma)$, a mapping from the natural numbers $\mathbb{N}$ to the partitions of our set of strings $\Sigma$ as follows.
\begin{enumerate}
\item 
$\theta(0) = \{\{s_1\}, \{s_2\}, \dots, \{s_n\}\}$. (i.e. the bottom level contains all singletons).
\item For each positive integer $i \leq l$, $\theta(i)$ contains the equivalence classes of $\Sigma$ under $r_{i}$. 
\item For $i > l$, $\theta(i)  = \{\Sigma\}$.
\end{enumerate} 
Simply put, in partition $\theta(i)$ each block contains strings that agree in the first $l-i$ positions.
To check that $\theta$ is indeed a dendrogram, we need to make sure $\theta(i)$ is a refinement of $\theta(i+1)$ for each $i$.
This can also be verified straightforwardly from the definitions of $r_i$: for any two strings $s$ and $s'$ such that $s \sim_{r_i} s'$, $s$ and $s'$ must share the same $l-i$ symbols by definition. Hence they must also share the same $l-i-1$ symbols (since $l-i-1 < l-i$) and it follows $s \sim_{r_{i+1}} s'$.
Our trip trie is in fact identical to this dendrogram; each level in our trie is a partition of all trips where each node is a block that contains all trips with the particular prefix denoted by the node. The leaf nodes correspond to $\theta(0)$. The level right above the leaves corresponds to $\theta(1)$, etc. and the root is $\theta(l+1)$.

\textbf{Metric on strings}.
Instead of constructing the trie as described, we could define a metric on the strings such that the standard single linkage agglomerative hierarchical clustering \cite{Carlsson2010,Jain1988} using this metric would result in the same dendrogram described above, that is, identical to our trip trie. 
Intuitively we are measuring the distance between two strings by the number of positions they differ where positions towards the start are weighted more. We want to weigh those positions to guarantee that two strings sharing \textit{only} the first position are still closer than another pair of strings that share \textit{all but the first} positions.
From this we define
\begin{equation}
d(s, s') = \sum_{i=1}^{l} 2^{(l-i)} \mathbb{I}[s^{(i)} \neq s'^{(i)}]
\end{equation}
where $\mathbb{I}$ is the 0/1 indicator function.

\begin{theorem}
$d$ is a metric on the space $\Sigma$.
\label{thm:metric}
\end{theorem}
\begin{proof}
It is straightforward to verify that $d(s,s) = 0$ and $d(s, s') = d(s', s)$. To prove that $d$ is indeed a metric on $\Sigma$, we show the triangle inequality: for any strings $s, s', s''$, $d(s, s'') \leq d(s, s') + d(s', s'')$.
Note by the definition of $d$ and the distributive property of multiplication, it suffices to show that for each $i$,
\begin{equation}
\mathbb{I}[s^{(i)} \neq s''^{(i)}] \leq \mathbb{I}[s^{(i)} \neq s'^{(i)}] + \mathbb{I}[s'^{(i)} \neq s''^{(i)}]
\end{equation}
We consider each of the two possible cases.
If $s^{(i)} = s''^{(i)}$, then the left side is 0 and the right side is either 0 (if $s^{(i)} = s'^{(i)}$, too) or 2 (if $s^{(i)} \neq s'^{(i)}$). Thus the inequality is satisfied.
On the other hand, if $s^{(i)} \neq s''^{(i)}$, then the left side is 1. Since $s^{(i)} \neq s''^{(i)}$, it is impossible that $s'^{(i)} = s^{(i)}$ and $s'^{(i)} = s''^{(i)}$ at the same time. Thus at least one of the two indicators on the right will output 1 and the right side will be at least 1; the inequality holds in general.
\end{proof}

\textbf{Remark}
To have agglomerative hierarchical clustering result in exactly the same dendrogram as our trie, we must allow more than two clusters to merge at each step (if the linkage distance between them is the same). 

\subsection{Run Time Complexity}
\label{subsec:runtime}
Here we give a brief analysis of the run time of our proposed algorithm.
The first step of spatial discretization defines a simple mapping from longitude/latitude coordinates to a finite symbol space. With our equal-sized rectangular grids, this mapping can be computed in constant time (for each coordinate) and thus the overall complexity of the discretization and the construction of string representation for all trips is $O(nl)$ where $n$ is the total number of trips and $l$ is the length of the longest string, which is determined by the duration of the trip and the chosen temporal resolution. 
From the pseudo-code in Figure \ref{fig:overall}, we can easily see that construction of the trie requires invoking \texttt{unique} $l$ times. Since straightforward implementation of \texttt{unique} takes $O(n)$ time, the overall complexity of building the tree is $O(nl)$ too. This is the same as the standard trie implementation using linked nodes. 
It is worth noting that the distributed implementation of prefix tree already exists in popular softwares (e.g. Apache HBase) and thus could scale to huge data sets.
In comparison standard hierarchical clustering algorithms require $O(n^2 \log n)$ time, which becomes prohibitively demanding when $n$ is as small as $\sim 10,000$. Later in our experiments, we build a trip trie of $ \sim 430$K trips (strings) in 100's of seconds.

\section{Other Uses of Trip Tries}
\label{sec:other}

The literature on tries and trees is significant and many useful properties have been derived. Here we describe how  to use two such properties: i) The ability to efficiently compute Levenshtein distances from tries and ii) Outlier score computation from trees.

\subsection{From Micro to Macro Clusters With Relaxed Distance Calculations}
The distance function that our method effectively uses (see equation \ref{equ:metric}) is quite strict. Two trips with overwhelmingly similar trajectories and speed may be placed in different clusters if the differences between the trips are towards the beginning of the trip. This may be desirable in some settings but in others more tolerance of these slight differences may be required.
Here we discuss how this can be achieved by exploiting that the prefix tree can be used to efficiently calculate the Levenshtein distance \cite{chou2003} between strings. Suppose we choose a level $l$ of our dendrogram we term all clusters at that level micro-clusters. In our experiments there are upwards of 10,000 such micro-clusters at level 20, each representing a unique path (i.e. trip type) through the space. We can group together these micro-clusters based on their Levenshtein distance. 

The Levenshtein distance between two strings is the number of operations (\texttt{insertion, deletion, substitution}) so that the two strings are the same \cite{levenshtein1966}. Therefore two strings $z_1 z_2 z_3 z_4$ and $z_2, z_2, z_3, z_4$ will be in different clusters in our hierarchial clustering but could be grouped together as their Levenshtein distance is just 1 (\texttt{substitution} at position 1). The Levenshtein distance is suitable for forming these clusters based on the micro-clusters as slight variations in trajectories can be overcome via the \texttt{substitution} operation and slight variations in speed can be overcome with the \texttt{insertion} and \texttt{deletion} operations. Using clustering objectives such as minimizing maximum cluster diameter will produce clusters with a strong semantic interpretation: a diameter of $q$ means all trips have at most $q$ differences in speed and route.

\subsection{Outlier Scores}
In this paper we mainly explore the interpretation and use of our trip tries as a clustering tool. However we introduce a different view on the trip trie here. A key to understanding this interpretation is that our discretization of space into a grid defines a state space where each state is represented by a symbol in the alphabet. Every root-to-leaf path is then a unique trip type and all root-leaf combinations represents all unique trips through the state space that occurred in our data set. Note this definition of unique trips (due to our distance metric) implies that trips starting and ending at the same location but using different routes and/or paces will be considered different.

The notion of Isolation Forest \cite{liu2008,liu2012} provides a method of identifying outliers as being those entities far from the root of the tree. We can use our trip tries to achieve a similar purpose. 
The frequency of a symbol (which represents a location) in the trie is an indication of its outlier score with respect to how many unique trip types involve it. If location $z_1$ occurs twice as much as $z_2$, then the former is involved in twice as many unique trips as the latter. This is visualized in Figure \ref{fig:from_sfo} for trips starting at San Francisco airport.
The depth of each symbol in the tree is an indication of how often it is used. If some measure of depth of $z_1$ is twice as large as $z_2$, then the former is much less prevalent at the beginnings of trips than the latter. This is visualized in Figure \ref{fig:loc_occur_order} where we consider first appearance in the tree though the mean level of occurrence could also have been used.

\section{Experiments}
\label{sec:experiment}
We have choosen experiments to demonstrate the usefulness of our proposed approach on a real world data set of freely avaialable GPS traces of taxis\footnote{Data set can be downloaded, after registration, at \texttt{http://crawdad.org/epfl/mobility/}.} \cite{cabspottingdata}.
In particular we focus on the following questions:
\begin{itemize}
\item Do results from a \texttt{MATLAB} agglomerative hierarchical clustering algorithm and our trie actually agree with each other? (verifying Theorem \ref{thm:metric}).
\item Is the distance metric we implicitly use useful? We address this by investigating properties of hierarchical clusterings built from different time periods to see if the insights found make sense (see Table \ref{tab:tree_stats}).
\item Can we explain clusters using our dendrogram's spatial and temporal interpretation? For example, where are the most common trip trajectories? Where do trips starting at a particular region go? 

\end{itemize}

We start with the description of our data set and the processing step we use to extract information of trips for reproducibility\footnote{All our code will be made publicly available for replication of our results upon acceptance.}.

\subsection{Data Description and Preparation}
\label{subsec:data}
The raw data set contains GPS traces of 536 taxis from Yellowcab in San Franciso bay area during a 24-day period in 2008. Each trace file corresponds to one taxi and consists of recordings of the latitude, longitude, a customer on/off flag and the time of recording (i.e. \texttt{37.75134, -122.39488, 0, 1213084687}). We generate a trip by searching for contiguous values of '1' for the customer-on flag for each taxi. Overall a total of $438145$ trips are extracted.
Since the vast majority of the taxi trips are short in duration and we are more interested in analyzing local traffics, we decide to study those trips whose trip time is $\leq 30$ minutes. These trips in fact account for $98.3\%$ of all trips extracted from the data set.
In addition, we pick our temporal resolution to be 1 minute. That is, in our string representation, consecutive symbols are the regions recorded at 1 minute apart. Assuming an average speed of $\sim 35$ miles per hour in the city, a taxi moves $\sim 0.583$ miles every minute. Accordingly we partition our modeled space into a $100 \times 100$ grids where each grid cell has a corresponding geographic dimension of $0.55$ miles $\times 0.54$ miles, which means adjacent symbols in a string are more likely to be different.

\subsection{Question 1: Equivalence of MATLAB Hierarchical Clustering and Our Approach}
Here we verify our claim (in section \ref{sec:property}) that the trie constructed from our approach is the same as the dendrogram output by standard agglomerative hierarchical clustering with metric defined in equation \ref{equ:metric}.
We constructed our trie and compared it with the results from the built-in single linkage clustering from MATLAB (i.e. \texttt{linkage}, followed by \texttt{dendrogram}) and check their equivalence. 
For any given level we can verify the two clusterings from our trie and the \texttt{MATLAB} dendrogram are identical up to reordering/relabeling with Algorithm \ref{alg:verify_same_cluster}.
If this is true for all levels, then we conclude the equivalence between our trie and the \texttt{MATLAB} dendrogram.
Note standard agglomerative clustering needs to compute and update distances between each pair of instances, which is prohibitive when the number of instances is large as in our case (in fact, standard agglomerative clustering cannot even handle moderate-sized data sets; see \cite{Gilpin2013cikm}). 
Therefore, we draw a random subset of 1000 trips and compare our trip trie built on them with the dendrogram output by agglomerative hierarchical clustering with string metric in equation \ref{equ:metric}. We repeat this for 10 random samples and in each case the dendrograms produced by MATLAB and our method are identical.

\begin{algorithm}[!h]
\begin{small}
\KwIn{Two lists $c_1, c_2$ of length $k_1$ and $k_2$, respectively, indicating which cluster an instance is in $1, 2, \ldots$}
\KwOut{\textit{true} if $c_1$ and $c_2$ are identical up to relabeling; \textit{false} otherwise}
$k_1 \leftarrow \texttt{max}(c_1)$\;
$k_2 \leftarrow \texttt{max}(c_2)$\;
\If{$k_1 \neq k_2$}{return \textit{false}\;}
\For{$i \leftarrow 1, \dots, k_1$}{
	$\textit{index} \leftarrow c_1 = i$\;
	$\textit{label} \leftarrow c_2(index)$\;
	\If{$\texttt{length}(\texttt{unique}(label)) \neq 1$}{return \textit{false}\;}
}
return \textit{true}\;
\end{small}
\caption{Verify whether or not two clusterings are identical up to relabeling of cluster numbers.}
\label{alg:verify_same_cluster}
\end{algorithm}

\subsection{Question 2: Usefulness of Hierarchies - Comparing Different Clusterings}
\label{subsec:tree_stat}
Here we demonstrate the usefulness of the hierarchies we build by exploring the properties of different dendrograms built from different groups of data. We categorize trips into distinct groups and construct a dendrogram for each of the groups. 
We are interested to know if hierarchical clusterings (dendrograms) built from distinct subsets of the trips exhibit any differences in their properties and if these differences make sense.
We use the start times of  trips to extract four subsets of the data: \emph{day peak} (5 AM - noon), \emph{night peak} (3 PM - 10 PM), \emph{weekdays} and \emph{weekends}. Note that these subsets are not the same in size and not mutually disjoint either hence we report the average of these properties. 



\textbf{Measures of Diversity.} 
One interesting characteristics of the trips is their diversity which can be measured in two ways. The branching factor of the dendrogram at each level provides a measure of \textit{dispersion} of the trips: if a dendrogram has large branching factors on average, then trips have more different trajectories in general. In Table \ref{tab:tree_stats} (line 2 and 3) we see that weekend trip dendrogram have a slightly higher average branching factor. 
A second measure is the number of times a region appears in the dendrogram divided by the number of clusters. If this number is significantly higher for one dendrogram than another then it means there are more diverse routes in the former dendrogram. Table \ref{tab:tree_stats} (line 4) shows night time has more diverse routes than day time and weekdays more than weekends.
Finally, Table \ref{tab:tree_stats}  (lines 6-7) report the average number of trips per cluster and we find as expected there are bigger clusters for the nighttime and weekday trips. 

\begin{table*}[!th]
\centering
\begin{small}
\begin{tabular}{|l|c|c|c|c|}
\hline
& Day peak & Night peak & Weekdays & Weekends\\
\hline
Total number of trips &  96582   &   161355   &   282545   &   148159 \\
Level-wise Average branching factor & 1.3234  &  1.3352  &  1.3531 &  1.3361  \\
Level-wise Average branching factor (first 11 levels) & 1.8802 &   1.9108  &  1.9598  &  1.9145 \\
Average number of clusters per region &  191.9070 & 251.0469 & 304.4685 & 238.3682 \\
Average number of clusters per region (first 11 levels) &  89.0493 & 113.2751 & 151.1392 & 113.0555 \\
Average number of trips per cluster &  12.5902 &  19.1384 &  23.6316 &  16.8952 \\
Average number of trips per cluster (first 11 levels) & 33.3085 &  51.5930 &  64.1925 &  45.3436 \\
\hline
\end{tabular}
\end{small}
\caption{
Statistics of the dendrograms built from different categories of trips. In our categories, ``Day peak'' and ``Night peak'' include trips whose start times are from 5 AM to noon and from 3 PM to 10 PM, respectively.}
\label{tab:tree_stats}
\end{table*}

 
\begin{figure*}[!th]
\begin{center}
\subfigure[Start]{\includegraphics[width=0.25\textwidth]{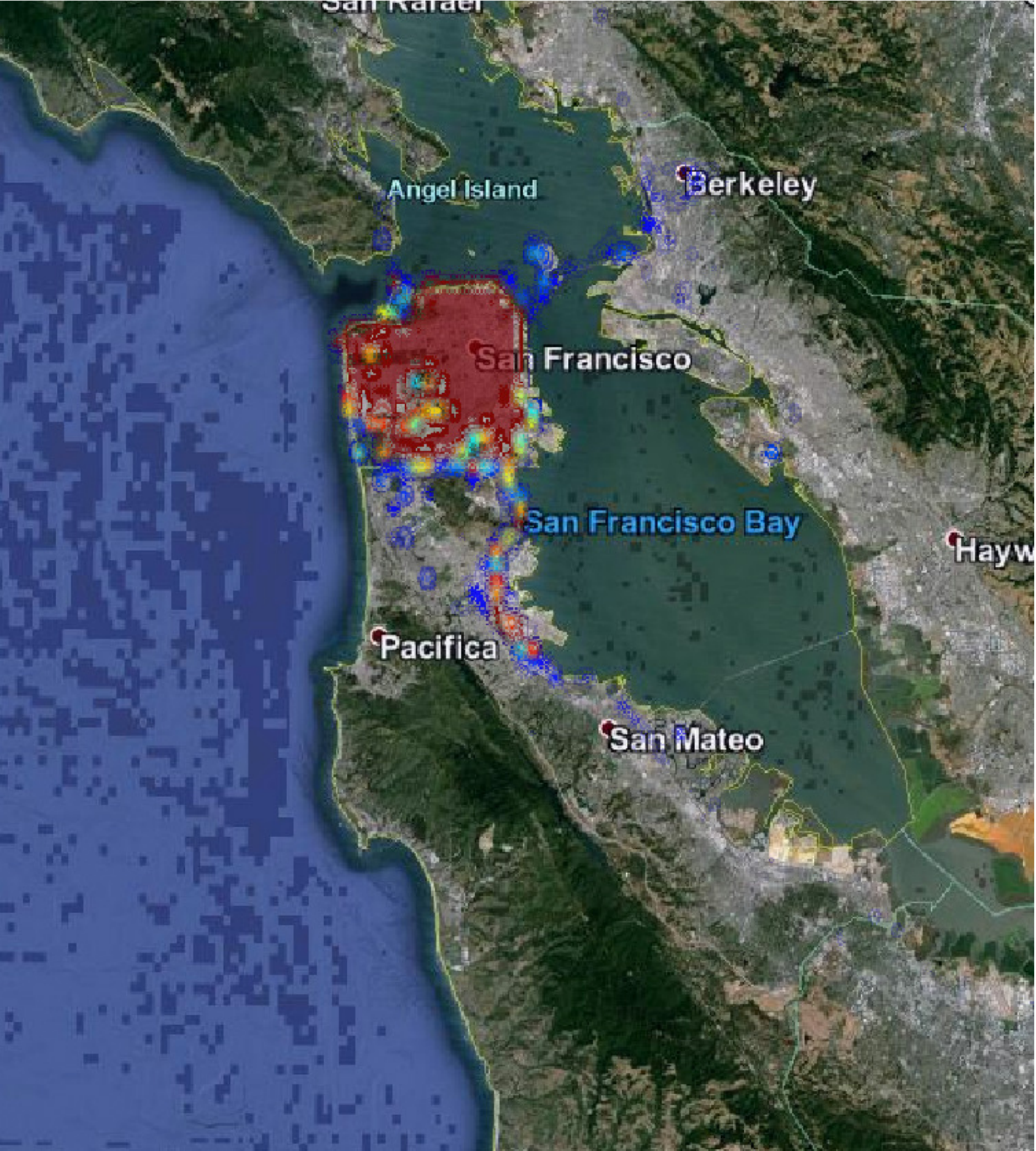}}\label{subfig:heatmap_start}
\quad
\subfigure[10th minute]{\includegraphics[width=0.25\textwidth]{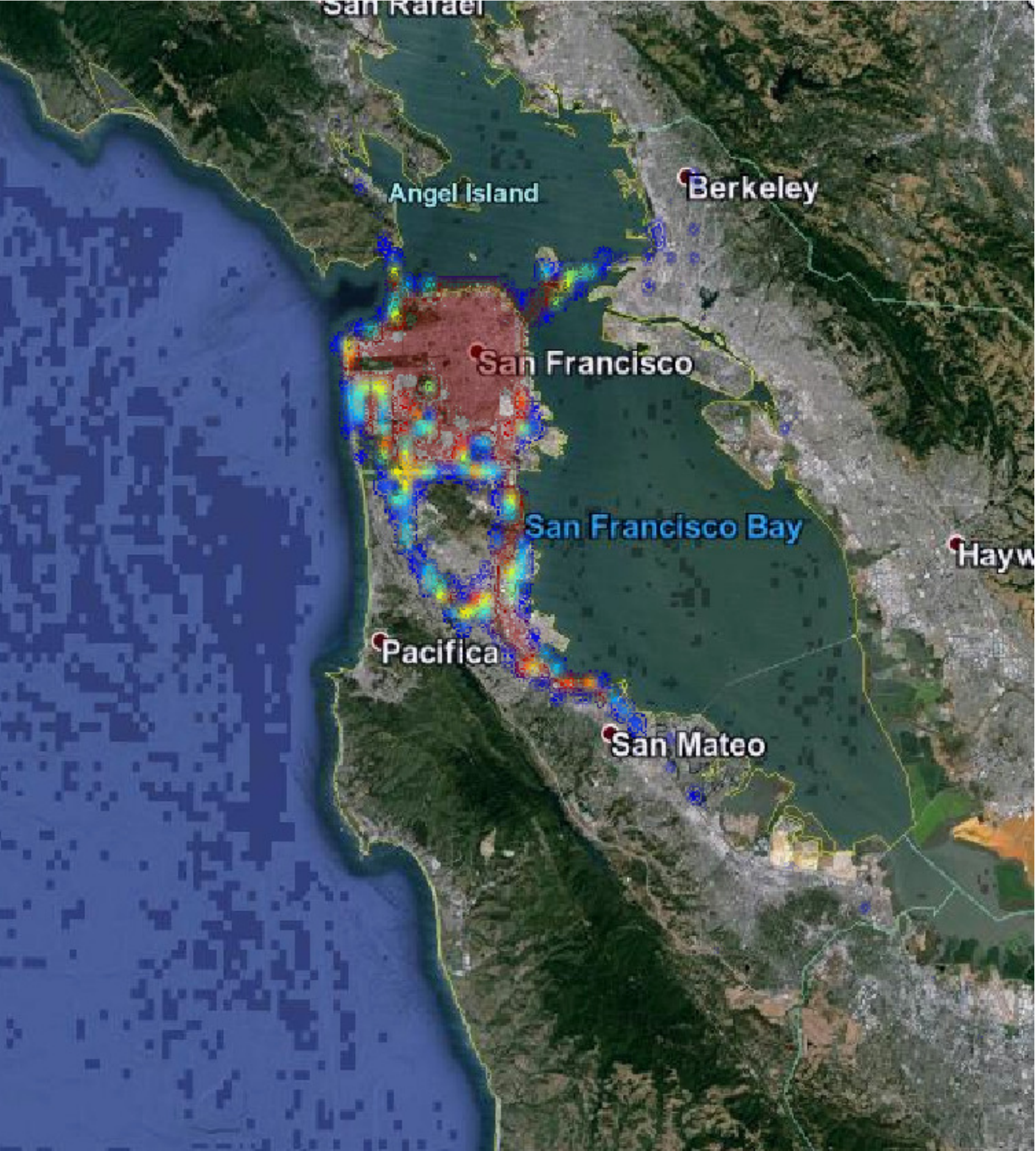}}\label{subfig:heatmap_10min}
\quad
\subfigure[20th minute]{\includegraphics[width=0.25\textwidth]{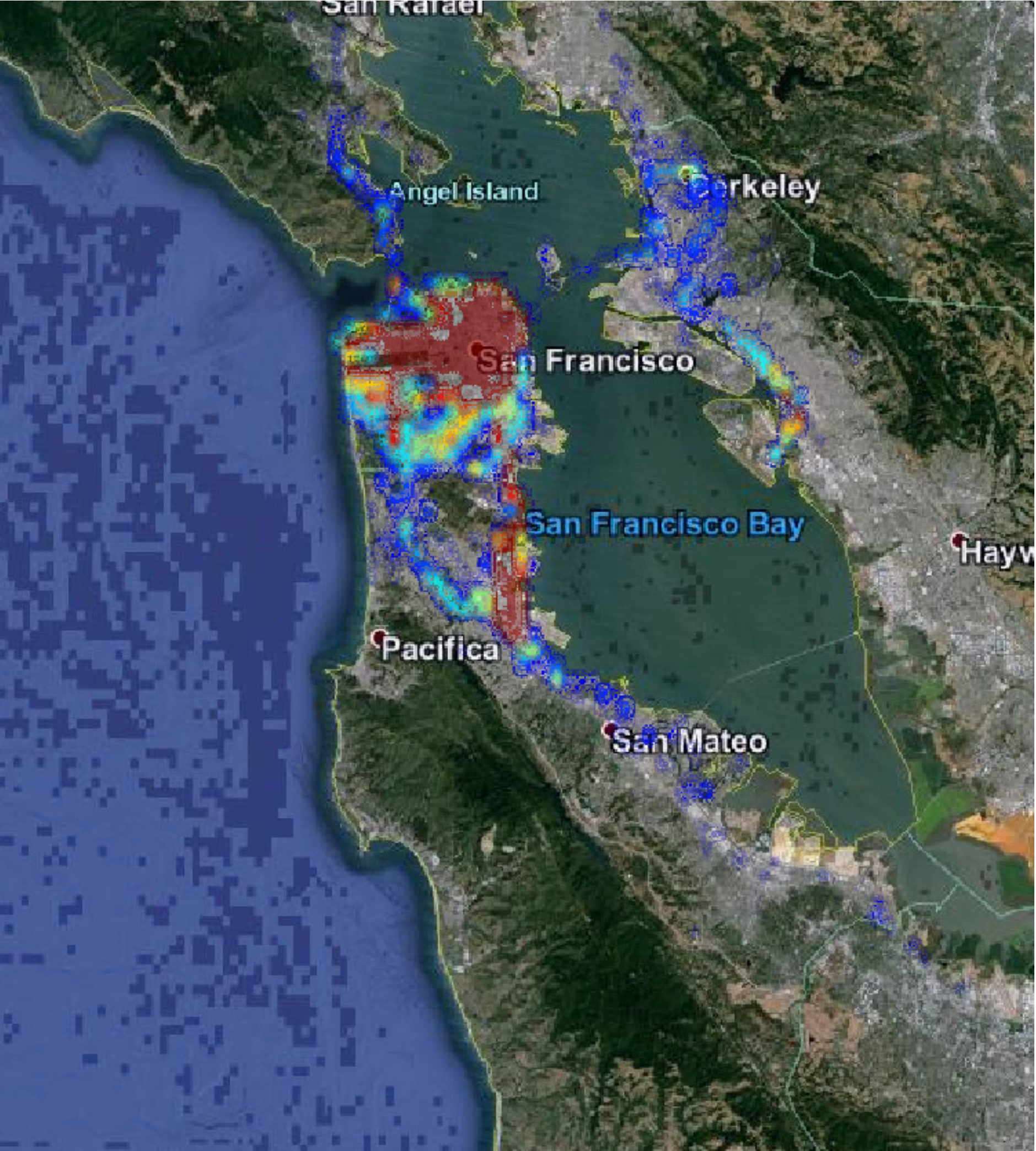}}\label{subfig:heatmap_20min}
\end{center}
\caption{Heat maps on locations of the trips at the start of trips (i.e. pickup), 10th minute and 20th minute after start.}
\label{fig:heatmap}
\end{figure*}

\begin{figure*}[!th]
\begin{center}
\subfigure[First cluster]{\includegraphics[width=0.25\textwidth]{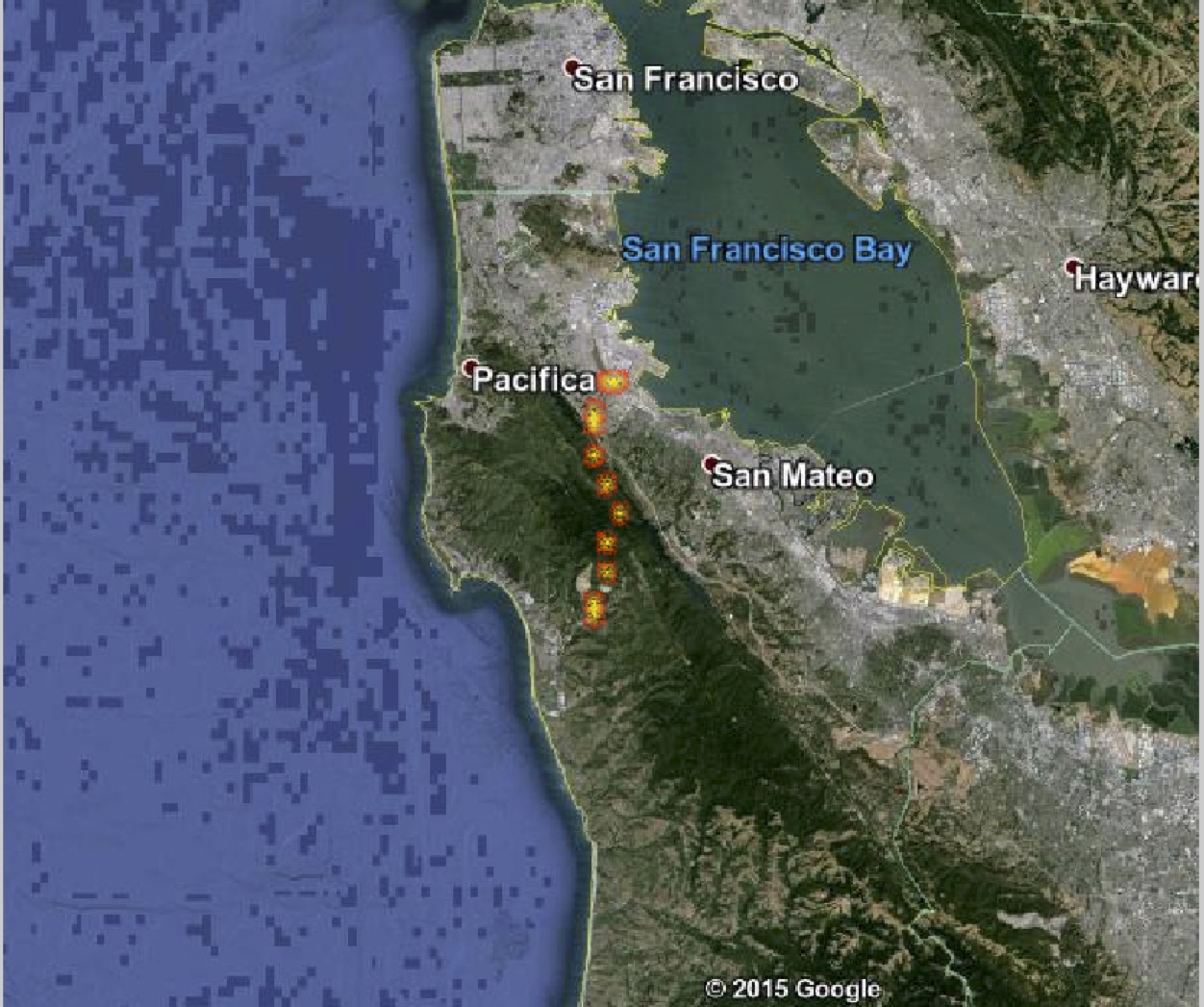}}\label{subfig:cluster1}
\qquad
\subfigure[Second cluster]{\includegraphics[width=0.25\textwidth]{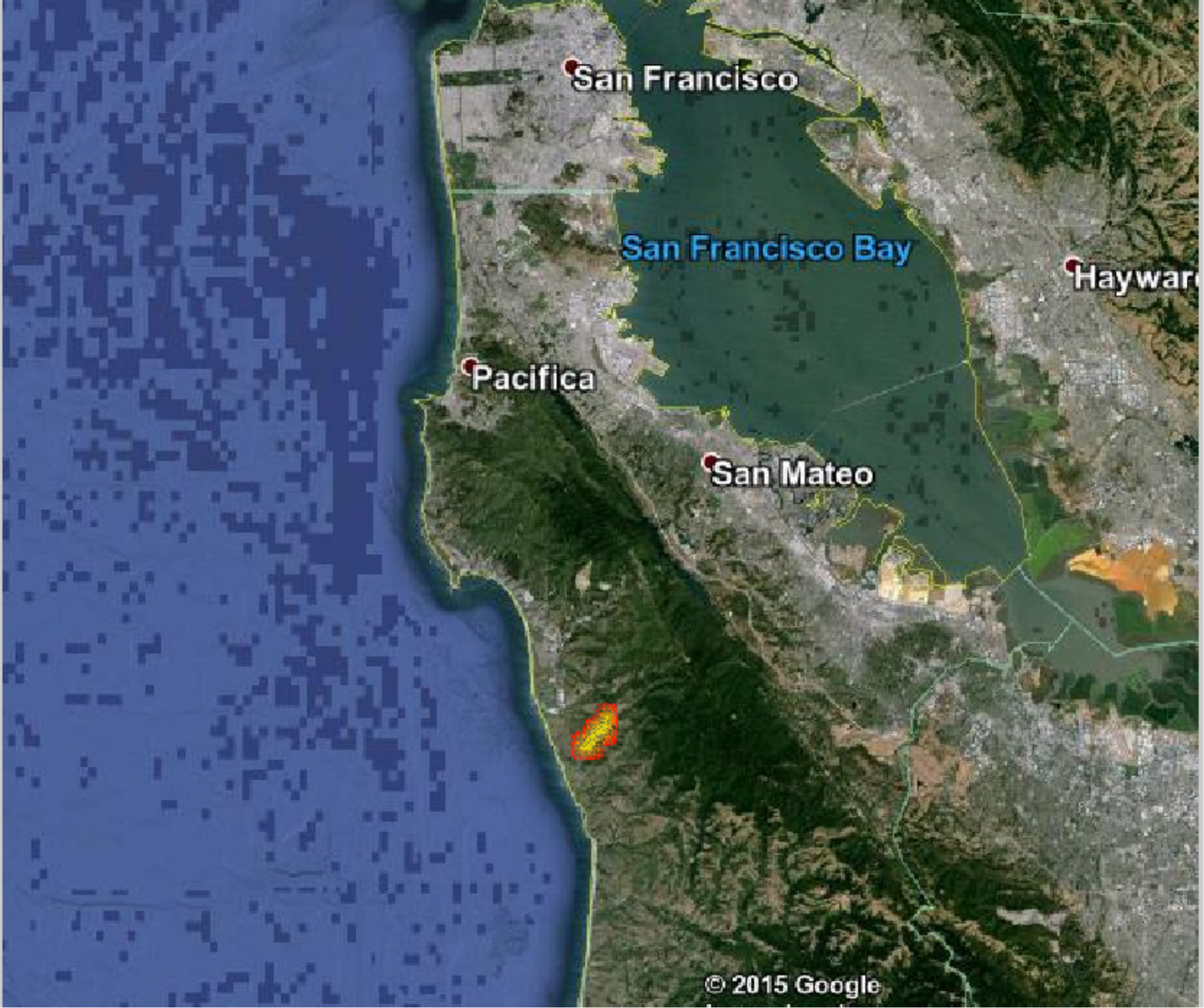}}\label{subfig:cluster2}
\qquad
\subfigure[Third cluster]{\includegraphics[width=0.25\textwidth]{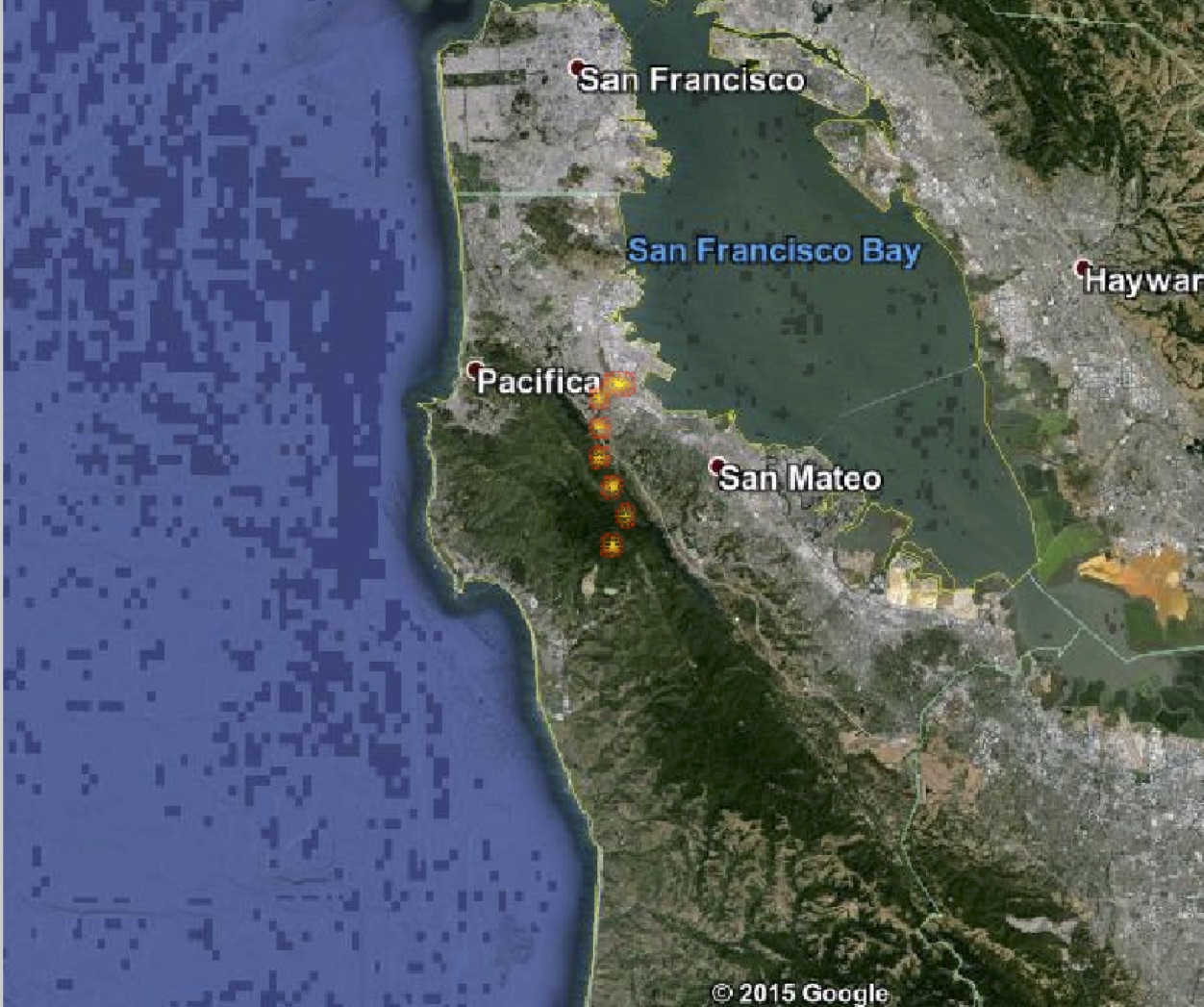}}\label{subfig:cluster3}
\end{center}
\caption{Three clusters at depth 11 in our dendrogram visualized as a trajectory. Compare with Figure \ref{fig:heatmap}.}
\label{fig:top_clusters}
\end{figure*}

\begin{figure*}[!th]
\begin{center}
\subfigure[Heat map of ``order of occurrence''. Notice that for the purpose of more intuitive display, We associate regions appearing early in trips with larger numbers (red) and regions appearing later with smaller numbers (blue).
]
{\includegraphics[width=0.25\textwidth]{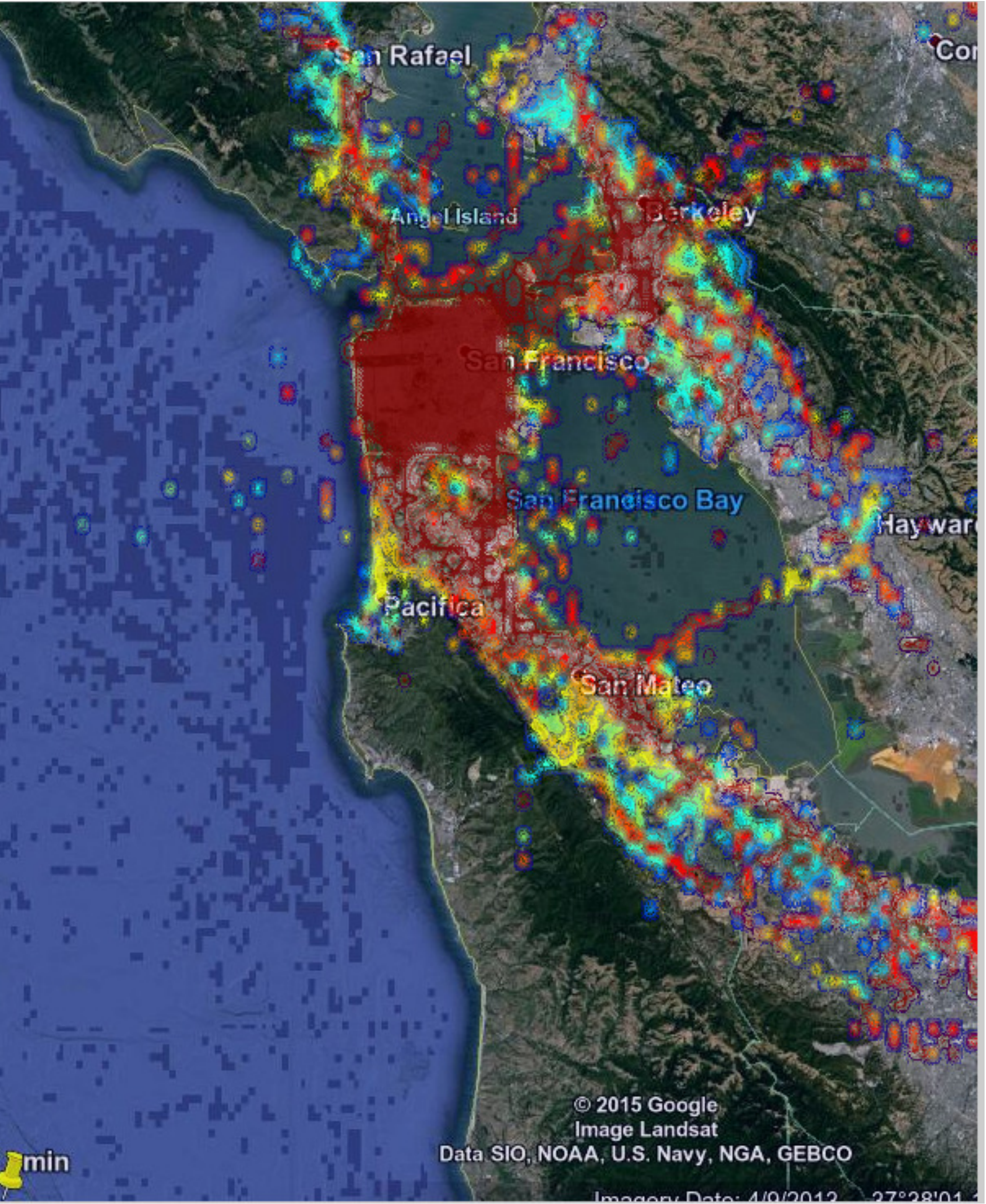}} \label{fig:loc_occur_order}
\qquad
\qquad
\subfigure[10 most frequent appearing regions (blue) at 10th minute after start for trips starting at the San Francisco airport (red).]{\includegraphics[width=0.3\textwidth]{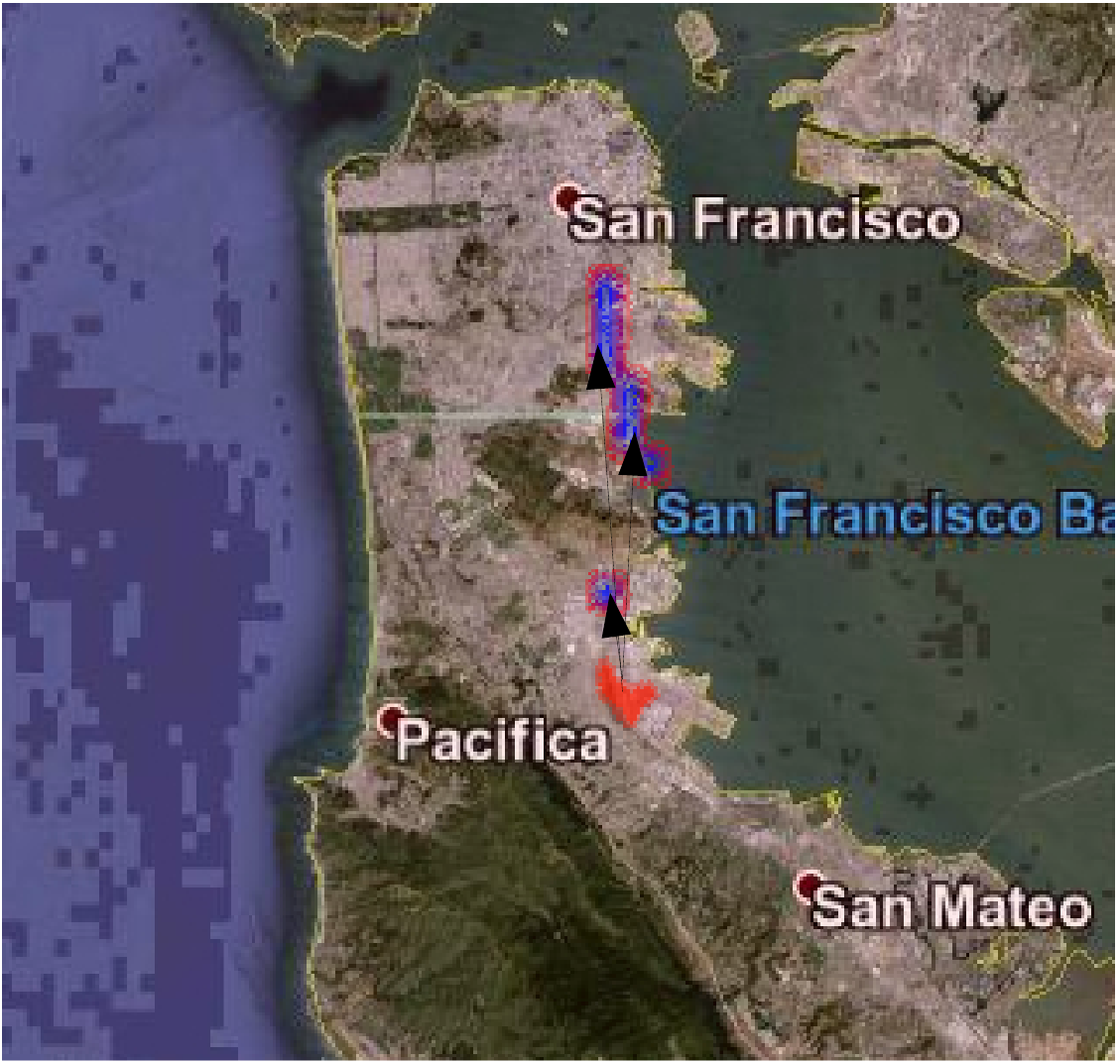}\label{fig:from_sfo}}
\end{center}
\caption{Some interesting uses of the trie to examine either information about the regions or a subset of trips.}
\label{fig:trie_more_use}
\end{figure*}

\subsection{Question 3: Temporal Interpretations of our Hierarchies}
\label{subsec:movement}
Each dendrogram has a natural temporal interpretation, the clustering at level $i$ is where the trips are at time step $i$. Here we explore finding high level movement patterns of the taxis using our hierarchy. 
Some natural questions regarding such movements are ``where do the trips start?'' and ``what paths are taken by the trips?''.

\textbf{Visualizing Frequencies of Regions/Symbols at Different Levels.} 
At each level of our dendrogram we will have a certain number of clusters each with a differing amount of trips. We can count the total number of trips through each distinct region/symbol across all clusters.
These counts can be represented as a heat map for a given time/level and these heat maps provide information on how taxis in general move about in the city. 
Figure \ref{fig:heatmap} shows 3 selected heat maps constructed from our dendrogram.
Figure \ref{subfig:heatmap_start}, \ref{subfig:heatmap_10min} and \ref{subfig:heatmap_20min} show the densities of taxis' locations at the start of the trips, 10th minute and 20th minute after the start, respectively.
From these figures, we can see that San Francisco downtown area always has traffic but there are general movement patterns radiating outwards from the area too. The most obvious interpretation is that these taxis carry customers from the more central area towards the north, east and south.



\noindent
\textbf{Visualizing Clusters as Trajectories.}
In our dendrogram each cluster at  level $i$ can be naturally interpreted as a partial trip/trajectory upto time $i-1$. Hence for any level we can easily visualize the most largest clusters as a trajectory. 
Figure \ref{fig:top_clusters} shows the top 3 largest clusters (partial trajectories) of our dendrogram at depth 11 (i.e. 10th minute after start of the trips). This is an interesting contrast to Figure \ref{fig:heatmap} which showed most trips originating in the downtown area. We find that after 10 minutes those trips have dissipated in so many directions that the most frequent trips do not include any originating in the downtown.

\textbf{Visualizing Regions across Clusterings.}
Often we are interested in the  order which regions appear in the trajectories. Are some regions more likely to be at the beginning of a trip or the end? 
For example the regions where people often start their trips will appear on the first level of our dendrogram (i.e. crudest clustering) and the places that only appear as dropoff points will appear at the leaf level (i.e. finest clustering). 
Accordingly we can construct a map of all regions and associate each region with the depth of the dendrogram in which it \emph{first} appears in \emph{any} cluster.
This can be efficiently computed from our dendrogram by finding the level which a symbol/region first appears. 
Figure \ref{fig:loc_occur_order} shows a heat map of this ``order of occurrence'' map. This map can be interpreted that those areas in red appear often at the start of a trip and those in blue appear towards the end of a trip.
We can see the regions in Downtown/Berkeley/Oakland appear in the crudest clusterings whereas regions farther away from road segments often occur in much finer clusterings only.

\textbf{Visualizing Refinement of a Clustering.}
Our dendrogram  models how a clustering (at a level) is refined. Some clusters may split into many new clusters at the next levels while some others split into fewer clusters and others even stay the same.
In our dendrogram, the split of a cluster in a refinement (i.e. next level) depends on the next regions the trajectory goes.
Often we are interested to look into the subset of trips that start from a particular region and see where they go. 
We can obtain such information by first choosing the cluster associated with the given start region (at level 1), and then following all the new clusters that are split from it in the refinements. This corresponds to the sub-tree (also a dendrogram) rooted at a particular cluster (at level 1) in the dendrogram.
Figure \ref{fig:from_sfo} shows the most frequently appearing regions at the 10th minute for those trips that start at the San Francisco airport. This could potentially be used as a predictive tool providing the distribution of end regions for all past trips starting at particular regions.

\begin{table*}[!th]
\centering
\begin{small}
\begin{tabular}{|c|c|}
\hline
Tasks & Time (seconds) \\
\hline
Extracting trips & $\sim 50$ \\
Constructing string representations & $\sim 650$ \\
Constructing trie & $\sim 60$ \\
Calculating trie statistics (section \ref{subsec:tree_stat}) & $\sim 0.5$ \\
Generating movement (or clusters) heat maps (section \ref{subsec:movement}) & $ 50 \sim 200$ \\
Generating region occurrence heat map (section \ref{subsec:movement}) & $\sim 0.5$ \\
\hline
\end{tabular}
\end{small}
\caption{Example run times for all parts of trie construction and various experiments on the entire data set of $\sim 430000$ trips.}
\label{tab:runtime}
\end{table*}

\subsection{Run Time}
\label{subsec:exp_runtime}
In earlier sections we show (by complexity) that our trip trie is much faster to construct over standard aggplomerative hierarchical clustering. Here we document the actual run time taken for each of the experiments described above in Table \ref{tab:runtime}.
All experiments were performed using MATLAB version 7 on Intel(R) Core(TM) i7-2630QM CPU @ 2.00GHz (no parallel computation was implemented and thus the presence of multi-cores is of little relevance).
The run times as presented in Table \ref{tab:runtime} are very reasonable for a data set of roughly half of a million instances. Given that our method scales linearly and the existence of modern parallel implementation of tries, we expect it is straightforward to apply our methods to much larger data sets of tens of millions with relatively little effort.

\section{Related Work}
Our work touches upon several areas of related work: hierarchical clustering,  spatial temporal mining, trajectory  mining and tree structures for trajectory data. To our knowledge the idea of building prefix trees to efficiently compute dendrogram structures is novel. The hierarchical clustering of large scale trajectory data sets is also novel as hierarchical clustering methods do not readily scale.

\label{sec:related}

\textbf{Hierarchical Clustering}
As we mentioned in the introduction, standard hierarchical clustering outputs a dendrogram and there is well known result that a dendrogram can be equivalently represented as an ultrametric through a canonical mapping \cite{hartigan1985,Carlsson2010}. Some work \cite{murtagh2008} also pointed out links between prefixes and ultrametrics when trying to increase ``ultrametricity'' of data through data recoding/coarsening.
Our current work makes the equivalence between a prefix tree and a dendrogram more formal, both analytically and empirically, and provide one metric between pairs of strigns by which standard hierarchical clustering outputs a dendrogram identical to our prefix tree.

\textbf{Spatial temporal mining:}
Spatial temporal data mining has been more recently studied  partially due to the emergence of cheap sensors that can easily collect vast amounts of data.
The spatial temporal nature of the data adds multiple challenges not handled by many classical data mining algorithms, such as discretization of continuous dimensions, non-independence of samples, topological constraints, visualization of the discovered results, and many more \cite{rao2012,andrienko2006,IJCAI09,AAAI10,KDD13,SDM13}.
Our  paper analyzes a particular form of spatial temporal data with a specialized data strucuture and we discuss some related work along this line.


\textbf{Trajectory data indexing and retrieval:}
There exists a body of literature deals with the storing, indexing and retrieval of large trajectory data sets \cite{dittrich2009,chakka2003,cudre2010}. Due to the spatial temporal nature, different static and dynamic data structures were proposed and explored, such as quadtree \cite{cudre2010}, 3D R-tree \cite{chakka2003,guttman1984}, etc. This line of work primarily focused on efficiently performing tasks common in databases, such as retrieving and updating data records.
Although data storage and updating are necessities of almost all data structures as well, our current work, on the other hand, is more focused on insightful and actionable discovery and summarization from a large collection of GPS trip data.

\textbf{Tree structures in spatial temporal data mining:}
Probabilistic suffix trees were used to mine outliers in sequential temporal database \cite{sun2006}. Tree-based structures were also studied and employed in optimizing queries and computations in large spatial temporal databases \cite{mouratidis2008,tao2003}.
Another work \cite{Monreale2009} attempted to predict the next locations of the given (incomplete) trajectories by first extracting frequent trajectory patterns (called T-patterns from \cite{giannotti2006}) and built a specialized prefix tree where the nodes are the (pre-determined) frequent regions and the edges are annotated with travel times. Next locations in trajectories are then predicted using association rules.
Our current work is meant to provide a exploratory analysis of the trips (i.e. clustering), rather than prediction, through a prefix tree structure that could potentially illuminate the users with new insights of the data. 

\section{Conclusion}
\label{sec:conclusion}
In this paper we propose a novel way to efficiently organize GPS trip data into hierarchy to gain high level actionable insights from such data.
We represent each trip symbolically as a string that contains its spatial temporal information and create a trie from these strings.
The trie partitions the trips at multiple granularities at different levels and can be shown to be equivalent to the output of standard agglomerative hierarchical clustering with a specific metric.
We discuss several uses of the trie including discovering traffic dynamics and characterizing outliers, and an empirical evaluation of our proposed approach on a real world data set of taxis' GPS traces demonstrates its usefulness.



One future work direction is allowing flexible dynamic changes to our trie while new trip records are collected and added. 
Examples of the changes include dynamically modifying the trie such as merging/splitting nodes (i.e. regions) when new trip records are collected or if information about densities of trips or road infrastructure are to be considered.


\bibliographystyle{ACM-Reference-Format}
\bibliography{ref}

\end{document}